\documentclass{article}

\usepackage[utf8]{inputenc}
\usepackage{amsmath, amssymb, amsthm, amsfonts}
\usepackage{fdsymbol}
\usepackage{bbold}
\usepackage{mathtools}
\usepackage{listings}
\usepackage{hyperref}
\usepackage{syntax}
\usepackage{tikz-cd}
\usepackage{stmaryrd}
\usepackage{listings}
\usepackage{enumitem}
\usepackage{subcaption}

\newtheorem{theorem}{Theorem}[section]
\newtheorem{proposition}[theorem]{Proposition}
\newtheorem{lemma}[theorem]{Lemma}

\theoremstyle{definition}
\newtheorem{definition}[theorem]{Definition}
\newtheorem{remark}[theorem]{Remark}

\usepackage{color}
\definecolor{keywordcolor}{rgb}{0.7, 0.1, 0.1}
\definecolor{tacticcolor}{rgb}{0.0, 0.1, 0.6}
\definecolor{commentcolor}{rgb}{0.4, 0.4, 0.4}
\definecolor{symbolcolor}{rgb}{0.0, 0.1, 0.6}
\definecolor{sortcolor}{rgb}{0.1, 0.5, 0.1}
\definecolor{attributecolor}{rgb}{0.7, 0.1, 0.1}
\definecolor{backgroundcolor}{rgb}{0.95, 0.95, 0.95}

\title{Formalizing Abstract Simplicial Complexes \& Stellar Subdivisions in Lean}

\newcommand\blfootnote[1]{%
  \begin{NoHyper}%
  \renewcommand\thefootnote{}\footnote{#1}%
  \addtocounter{footnote}{-1}%
  \end{NoHyper}%
}

\author{
  Garett Cunningham \\ 
  \small University of Connecticut \\ 
  \small garett.cunningham@uconn.edu
  \and
  Daniel Zach \\ 
  \small Universit\"{a}t Regensburg \\ 
  \small Daniel.Zach@stud.uni-regensburg.de
  \and
  Stefan Friedl \\
  \small Universit\"{a}t Regensburg \\
  \small stefan.friedl@mathematik.uni-regensburg.de
}

\begin{document}

\maketitle
\begin{abstract}
  The theory of simplicial complexes is a cornerstone of topology, offering a sophisticated tool for computing invariants. We present a formalization of abstract simplicial complexes and stellar subdivisions in the Lean proof assistant. We adopt a purely combinatorial framework in order to provide a cohesive foundation for studying the theory of stellar subdivisions as seen in many contexts of combinatorial topology. In particular, we provide formalizations of morphisms between abstract simplicial complexes; several crucial constructions and operations on complexes, such as links and joins; and perform a comprehensive study of how stellar subdivisions interact with these operations. We state and prove a number of identities commonly used in the study of triangulated manifolds, such as deriving equivalences between links in an abstract simplicial complex $K$ and in a stellar subdivision $\sigma_s K$, including results with no references in the standard literature. To our knowledge, this is the first formalization of stellar subdivisions in any proof assistant.
\end{abstract}\blfootnote{The authors acknowledge support from the German-U.S. Fulbright Commission and CRC 1085 ``Higher Invariants'' at Universit\"{a}t Regensburg, funded by the SFB.}

\section{Introduction}
The theory of simplicial complexes finds its roots in some of the cornerstone results of premodern geometry, beginning with the study of polyhedra in 3-dimensional Euclidean geometry. Perhaps the first topological result is the Euler characteristic and its independence from the choice of triangulation for a surface. The notion of a simplicial complex generalizes and subsumes this study, having grown to become an indispensable tool for algebraic topology, largely due to its quality of being especially amenable to computational techniques. The result is an array of methods for defining and computing topological invariants, such as simplicial homology and genus in the most elementary cases. Moreover, simplicial structure can be heavily leveraged to study the homotopy type of a space by gaining access to discretized versions of standard tools---for example, discrete Morse theory and simplicial collapse algorithms~\cite{kozlov2008combinatorial,kozlov2021organized}. Likewise, a specialization of spectral sequences computes higher homotopy groups in the simply connected case~\cite{sergeraert1994computability}, as implemented in Kenzo~\cite{sergeraert1999kenzo}.

Classically, the field of combinatorial topology has made a rigorous program of extending and refining this toolkit. The resulting efforts have received wide applications, especially in the modern era of computers. Topological data analysis often leverages representations of data as simplicial complexes to compute useful metrics or invariants. For example, persistent homology~\cite{edelsbrunner2008persistent} provides information regarding the structure of a dataset that is independent of certain biases in sampling. The software package Regina~\cite{regina} is designed for explicit computations in low-dimensional topology based on a triangulation of a space (e.g., a 3-manifold or knot), allowing preprocessing by simplifying a triangulation.

Concretely, this is done using \emph{bistellar} and \emph{stellar subdivisions}. These have the advantage of being especially pleasant for algorithms or combinatorial description. Stellar subdivisions have classically been a common choice for a ``standard'' subdivision~\cite{glaser1970geometrical, hudson1969piecewise, rourke2012introduction, seifert2004lehrbuch}. As such, they have become an essential tool in combinatorial topology---for example, in the classical proof of Newman's theorem~\cite{cohen1968proof}. Theorems of Alexander~\cite{alexander1930combinatorial} and Pachner~\cite{pachner1991pl} show that (triangulations of) manifolds are homeomorphic if and only if they are related by a sequence of stellar or bistellar moves, respectively. A direct combinatorial proof of the equivalence was given in a thesis of Hannes~\cite{friedl2024pachner}.

The traditional view defines simplicial complexes as structures embedded in Euclidean space, $\mathbb{R}^n$, via barycentric coordinates. This has the immediate benefit of strengthening ties to topological spaces. Moreover, this construction leads to a straightforward theory of general subdivisions, of which stellar subdivisions are a core representative. Indeed, to study general subdivisions, it suffices to study stellar subdivisions~\cite{glaser1970geometrical, hudson1969piecewise}. However, reliance on the ambient space adds additional data that makes computations increasingly untenable. As a result, one wishes to study the more abstract, combinatorial properties of simplicial complexes free of an ambient space, often called an \emph{abstract simplicial complex}. With this view, stellar subdivisions are equally easy (if not easier) to study.

Abstract simplicial complexes have received attention in other proof assistants, with existing formalizations in Rocq~\cite{heras2011incidence} and ACL2~\cite{aransay2021simplicial}. These projects were directed toward explicit computational applications---persistent homology in the former case~\cite{heras2013computing, heras2012towards} and Alexander duals in the latter~\cite{aransay2022formalizing}. We highlight this as a rather strange gap in Lean's development, given the long-standing interest in other communities. Indeed, \lstinline|mathlib|~\cite{mathlib2020} contains many formalizations that relate to simplicial complexes, but with no interconnection between them. We discuss how this motivated our choices in design throughout.\footnote{Since the preparation of this manuscript, {\footnotesize\texttt{mathlib}} accepted independent implementations of {\footnotesize\texttt{PreAbstractSimplicialComplex}} and {\footnotesize\texttt{AbstractSimplicialComplex}}. The former is identical to our formalization presented here and the latter includes an additional totality axiom. We refer the reader to Remark~\ref{rem:isomorphisms} and Section~\ref{sec:discussion} for a more thorough comparison to our work presented here.}

The combinatorial structure lends itself as a natural candidate for formalization. However, the highly geometric nature of the subject leads to much of the canonical literature appealing to intuition-focused arguments. One is given trust that the computations can be done and will succeed if checked. This has the pedagogical advantage of maintaining clarity of ideas, but leads to obvious gaps in logic. Certain identities with highly technical proofs may go increasingly unchecked by future generations, shifting toward the status of folklore instead of theorem. We see this work as a step toward cataloging results in a centralized and formally verified manner so that the details have been sufficiently worked out for the reader. Additionally, the task of patching published proofs led to new theorems that have, to our knowledge, no existing references in the standard literature.

Given the ubiquity of abstract simplicial complexes for computational applications, we present initial legwork toward formalizing them in Lean 4~\cite{moura2021lean}---in particular, some foundational theory regarding the combinatorial properties, constructions of subcomplexes, and initial work on stellar subdivisions. Our selection of theorems for formalization are informed by future work to formalize combinatorial manifolds, in particular toward a future formalization of Alexander and Pachners' theorems. Our contributions are as follows:
\begin{itemize}
  \item A formalization of abstract simplicial complexes; including morphisms, plus certain operations and constructions thereon.
  \item A formalization of stellar subdivisions, including identities relating how they interact with objects defined above.
  \item Formal proofs of new theorems that, as far as the authors are aware, are not found in the standard literature.
\end{itemize}

To our knowledge this is the first formalization of stellar subdivisions in any proof assistant. Throughout, we emphasize a design philosophy that prioritizes preserving computability and compatibility with existing \lstinline|mathlib| content as far as reasonably possible. The former is inspired by our interests in simplicial complexes as a computational tool with direct algorithmic applications. We detour to clarify the latter and preview some of the main issues faced in development. The guiding principle is to choose the fewest assumptions necessary, so that later efforts to connect components in \lstinline|mathlib| may be as streamlined as possible.

Some minor attention has been given to the direct study of simplicial complexes via \lstinline|SimplicialComplex|. However, its state of development is minimal. We chose to isolate the combinatorial aspects of \lstinline|SimplicialComplex| under a new type \lstinline|AbstractSimplicialComplex|. As our goals are purely combinatorial in nature, this decision removed many roadblocks and considerably reduced the overhead in formalization. The existing \lstinline|mathlib| definition is as follows:\footnote{{\footnotesize\texttt{mathlib}} has since changed the downward closure condition to use {\footnotesize\texttt{IsRelLowerSet}}. We plan to refactor our code to match this in a future update.}
\begin{lstlisting}
variable (𝕜 E) [Ring 𝕜] [PartialOrder 𝕜] [AddCommGroup E] [Module 𝕜 E]
structure SimplicialComplex where
  faces : Set (Finset E)
  empty_notMem : ∅ ∉ faces
  indep : ∀ {s}, s ∈ faces → AffineIndependent 𝕜 ((↑) : s → E)
  down_closed : ∀ {s t}, s ∈ faces → t ⊆ s → t ≠ ∅ → t ∈ faces
  inter_subset_convexHull : ∀ {s t}, s ∈ faces → t ∈ faces →
    convexHull 𝕜 ↑s ∩ convexHull 𝕜 ↑t ⊆ convexHull 𝕜 (s ∩ t : Set E)
\end{lstlisting}

Notably, the definition depends on the types $E$ and $\mathbb{k}$, where $E$ is a module over the partially ordered ring $\mathbb{k}$. This gives a minimal setting needed to define a simplicial complex as existing in an ambient $\mathbb{k}$-affine space. Formally, a simplicial complex is defined to be the convex hulls of its faces glued together nicely (\lstinline|inter_subset_convexHull|) with no degeneracies (\lstinline|indep|). For example, by setting $\mathbb{k} := \mathbb{R}$ and $E := \mathbb{R}^n$, one recovers the usual embedding in $n$-dimensional Euclidean space. We chose to drop these geometric conditions and highlight where potential overhead is generated as relevant, with a full discussion reserved for Section~\ref{sec:discussion}.

One final thematic difficulty is the dependence of definitions on some type. In particular, a simplicial complex (whether abstract or geometric) implicitly comes with the attached data of a labeling of its vertices. This is, primarily, a fragment of how Lean chooses to define sets as formulas $E \to \text{\lstinline|Prop|}$ for some type $E$. Therefore, we pay particular attention to how this data is affected by operations and maps, leading to a necessary theory of proper typecasting between simplicial complexes. Joins become the first roadblock motivating this level of care.

Throughout this paper, we do not assume familiarity with simplicial complexes or their applications. Indeed, we begin with the basic definition and work from the ground up. We refer the intrigued reader to some of the many excellent introductory resources on algebraic topology~\cite{friedl2019algebraic, munkres2018elements, sher2001handbook, spanier2012algebraic} and combinatorial topology~\cite{glaser1970geometrical, hudson1969piecewise, kozlov2008combinatorial, rourke2012introduction} for a deeper study. Definitions are presented in their Lean form and theorems are presented in their natural language form. All claims are formalized with the respective Lean names attached. Our code is publically available via our GitHub repository,\footnote{\url{https://github.com/not-gary/pachner}. All code associated with this work was developed without the aid of generative AI tools.} currently totaling $\sim$14,000 lines of code.

Section~\ref{sec:simplicial-complexes} presents a zoo of definitions and theorems describing the properties of abstract simplicial complexes and simplicial maps, as well as important subcomplexes for consideration; such as links, joins, and boundaries. Section~\ref{sec:stellar-subdivisions} extends the discussion to stellar subdivisions. In particular, we focus discussion on formalizing the notion of \emph{stellar equivalence} of simplicial complexes and prove a series of identities primarily describing the interaction of stellar subdivisions with the subcomplexes defined in Section~\ref{sec:simplicial-complexes}. Throughout, we take cursory detours to highlight design choices and complications encountered during development. Section~\ref{sec:discussion} is dedicated to a more thorough discussion thereof, with a view toward possible future work and an analysis of related contemporary work.

\section{Simplicial Complexes}\label{sec:simplicial-complexes}

\begin{figure}
  \centering
  \includegraphics[width=0.7\linewidth]{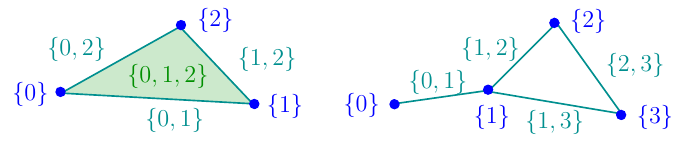}
  \caption{Examples of simplicial complexes with set-based encoding shown.}
  \label{fig:simplicial-complexes}
\end{figure}

The notion of simplicial complexes can be seen as stemming from the process of triangulating topological spaces. Informally, a simplicial complex is composed of $n$-dimensional faces (possibly none), inductively closed under its $(n - 1)$-dimensional edges for all $n \geq 1$. A point, or $0$-face, is called a \emph{vertex}. One can view simplicial complexes as a higher dimensional analog of graphs or polyhedra, which are special cases. In this more general setting, we are allowed to glue pieces of differing dimensions together. Figure~\ref{fig:simplicial-complexes} shows a visual example.

One method to encode the downward closure property is to label the vertices of a simplicial complex. An $n$-dimensional face is then defined as the set of $n + 1$ vertices that it contains (cf. Fig.~\ref{fig:simplicial-complexes}). In this manner, a simplicial complex is formally a collection of finite sets (faces) that is closed under the subset relation. To distinguish this from the geometric formulation, we call this an \emph{abstract} simplicial complex.
\begin{lstlisting}
structure AbstractSimplicialComplex (E : Type _) where
  faces : Set (Finset E)
  empty_notMem : ∅ ∉ faces
  down_closed : ∀ {s t}, s ∈ faces → t ⊆ s → t ≠ ∅ → t ∈ faces
\end{lstlisting}

In the sequel, we drop the adjective ``abstract,'' so long as the context is clear, and use $s \in K$ to mean ``$s$ is a face of $K$.'' We use $V(K)$ to denote the set of vertices of $K$. One may define this in a few equivalent ways.

\begin{definition}[{\small\texttt{AbstractSimplicialComplex.vertices}}]
  \[
    V(K) := \{x \mid \{x\} \in K\}.
  \]
\end{definition}

\begin{proposition}[{\small\texttt{vertices_eq}}, {\small\texttt{vertices_setOf}}]
  \[
    V(K) = \bigcup_{s \in K} s = \{x \mid \exists s \in K,\; x \in s\}.
  \]
\end{proposition}

In our definition of \lstinline|AbstractSimplicialComplex|, we retain the axioms of \lstinline|SimplicialComplex| that relate to the combinatorial data, namely \lstinline|empty_notMem| and \lstinline|down_closed|, and remove the geometric conditions from \lstinline|SimplicialComplex|, \lstinline|indep| and \lstinline|inter_subset_convexHull|. As a consequence, we no longer need the additional type $\mathbb{k}$ nor the assumptions in the type signature. Therefore, we define an \lstinline|AbstractSimplicialComplex| for just any base type $E$. If some ring $\mathbb{k}$ is provided such that $E$ is a module over $\mathbb{k}$, then one can retrieve an instance of \lstinline|SimplicialComplex|.
\begin{lstlisting}
variable (𝕜 E) [Ring 𝕜] [PartialOrder 𝕜] [AddCommGroup E] [Module 𝕜 E]
def AbstractSimplicialComplex.ofGeometric (K : SimplicialComplex 𝕜 E) :
  AbstractSimplicialComplex E :=
    ⟨K.faces, K.empty_notMem, K.down_closed⟩

def SimplicialComplex.ofAbstract
    (K : AbstractSimplicialComplex E)
    (indep : ∀ {s}, s ∈ K.faces → AffineIndependent 𝕜 ((↑) : s → E))
    (inter_subset_convexHull : ∀ {s t}, s ∈ K.faces → t ∈ K.faces →
      convexHull 𝕜 ↑s ∩ convexHull 𝕜 ↑t ⊆
        convexHull 𝕜 (s ∩ t : Set E)) :
  SimplicialComplex 𝕜 E :=
    ⟨K.faces, K.empty_notMem, indep, K.down_closed,
      inter_subset_convexHull⟩
\end{lstlisting}
Hence, any results that hold for \lstinline|AbstractSimplicialComplex| will also hold for \lstinline|SimplicialComplex|. Converting back and forth yields the original complex (e.g., \lstinline|ofAbstract_ofGeometric_id|).

We carry over the same constructions that exist for \lstinline|SimplicialComplex| as applicable. However, there is a number of basic notions missing that we develop. First, \lstinline|mathlib| lacks definitions for unions and intersections of complexes.
\begin{lstlisting}
def SimplicialUnion (K L : AbstractSimplicialComplex E) :=
  ⟨K.faces ∪ L.faces, _, _⟩
def SimplicialInter (K L : AbstractSimplicialComplex E) :=
  ⟨K.faces ∩ L.faces, _, _⟩
\end{lstlisting}
This induces instances of \lstinline|instHasUnion| and \lstinline|instHasInter| respectively to take advantage of the usual set-theoretic notation for unions and intersections. Likewise, the notion of a \emph{subcomplex} derives an instance of \lstinline|instHasSubset|.
\begin{lstlisting}
def IsSubcomplex (K L : AbstractSimplicialComplex E) : Prop :=
    K.faces ⊆ L.faces
\end{lstlisting}

Second, we define a notion of \emph{dimension} for both a face (via its cardinality) and, by extension, a simplicial complex.
\begin{lstlisting}
def face_dim (s : Finset E) : ℤ := Finset.card s - 1
def AbstractSimplicialComplex.dim
    (K : AbstractSimplicialComplex E) [Fintype K.faces] : ℤ :=
  Finset.max' ((Finset.image face_dim K.faces.toFinset) ∪ {-1}) _
\end{lstlisting}
We highlight two qualities of this definition. First, the usage of \lstinline|Finset.max'| returns an integer rather than an element of type \lstinline|Option| $\mathbb{Z}$, so long as the set is nonempty. We  assign the empty complex $\bot$ dimension $-1$ as convention. Second, this requires that $K$ be finite in order to extract a \lstinline|Finset| for calculation. In practice, one is often concerned with finite (i.e., compact) simplicial complexes and explicit calculations with them. We found this useful for some dimension preservation results.

\begin{remark}\label{rem:graphs}
  A special case of simplicial complexes is graphs. Indeed, one can view undirected graphs as 1-dimensional simplicial complexes. In Lean, a graph is called a \emph{simple graph} and is defined with vertices of arbitrary type $V$.
  \begin{lstlisting}
structure SimpleGraph (V : Type _) where
  Adj : V → V → Prop
  symm : Std.Symm Adj := by aesop_graph
  loopless : Std.Irrefl Adj := by aesop_graph\end{lstlisting}
  The case for multigraphs is represented by \lstinline|Graph|. Indeed, we implement a way to convert from a 1-dimensional abstract simplicial complex to either object (\lstinline|SimpleGraph.ofAbstract|, \lstinline|Graph.ofAbstract|).
\end{remark}

\subsection{Simplicial Maps}\label{sec:simplicial-maps}

We first define a \emph{simplicial map} from $K$ to $L$ as a function $f : E \to F$ such that every face in $K$ is mapped to a face in $L$:
\begin{lstlisting}
structure SimplicialMap
    (K : AbstractSimplicialComplex E) (L : AbstractSimplicialComplex F)
  where
    map : E → F
    is_simplicial : ∀ s, s ∈ K.faces → Finset.image map s ∈ L.faces
\end{lstlisting}

In particular, all values of \lstinline|map| outside the faces of $K$ are treated as ``junk values,'' since \lstinline|is_simplicial| imposes no condition on them. Note that the property of being a simplicial map depends on the complexes used for the domain and codomain. However, there are circumstances in which a map stays simplicial after we change either parameter. For example, restriction to a subcomplex yields a simplicial map (\lstinline|SimplicialMap.restrict|). Moreover, any function applied to a simplicial complex is simplicial onto its image, including restricting the codomain to the image (\lstinline|SimplicialImage|, \lstinline|SimplicialMap.ontoImage|). We informally denote the image of a simplicial map $f$ on $K$ by $f(K)$. Simplicial maps are also closed under composition. That is, if $f : E \to F$ is simplicial from $K$ to $L$ and $g : F \to G$ is simplicial from $L$ to $Z$, then $g \circ f : E \to G$ is a simplicial map from $K$ to $Z$ (\lstinline|SimplicialMap.comp|).

\begin{remark}\label{rem:category}
  This is sufficient to define the category of abstract simplicial complexes with morphisms given by simplicial maps. Indeed, the identity map $E \to E$ is always a simplicial map (\lstinline|idSimplicialMap|). We mark this as useful future work, wherein a subset of the authors have initiated active plans toward further development. Moreover, we note challenges to deriving an equivalent definition for \lstinline|SimplicialComplex| in Section~\ref{sec:discussion}.
\end{remark}

\subsubsection{Simplicial Isomorphisms}
Geometrically, an isomorphism of simplicial complexes, $K \cong L$, will preserve the simplicial structure, but potentially relabel the vertices. Formally stated:
\begin{lstlisting}
def IsSimplicialIso
    {K : AbstractSimplicialComplex E} {L : AbstractSimplicialComplex F}
    (f : SimplicialMap K L) : Prop :=
  ∃ g : SimplicialMap L K,
    (K.vertices).restrict (g ∘ f).map = (K.vertices).restrict id ∧
    (L.vertices).restrict (f ∘ g).map = (L.vertices).restrict id
\end{lstlisting}

The above definition is propositional and is agnostic to the particular choice of inverse $g : F \to E$. We also supply a bundled version that includes a specific choice of inverse simplicial map with proofs that they are mutually inverses on the vertex sets.
\begin{lstlisting}
structure SimplicialIso
    (K : AbstractSimplicialComplex E) (L : AbstractSimplicialComplex F)
  where
    toFun : SimplicialMap K L
    invFun : SimplicialMap L K
    left_inv : ∀ {x : E}, x ∈ K.vertices →
      invFun.map (toFun.map x) = x
    right_inv : ∀ {x : F}, x ∈ L.vertices →
      toFun.map (invFun.map x) = x
\end{lstlisting}
In line with Remark~\ref{rem:category}, \lstinline|CategoryTheory.Iso| uses bundled versions of isomorphisms. We provide both definitions to allow some freedom of choice for users.

\begin{remark}\label{rem:isomorphisms}
  Note that we only stipulate the bijectivity condition for the vertices of the simplicial complex. What the actual maps $f$ and $g$ do outside the set of vertices does not matter. One might define $f : E \to F$ to be a bijection of types. However, there arise examples where constructing an explicit inverse is tricky. For example, consider $E = F = \mathbb{N}$. Take
  \[
    K := \mathcal{P}(\{1, \ldots, n\}) \setminus \{\emptyset\}
  \] and the monus operation $f(x) := x \dotminus 1$. Then,
  \[
    f(K) = \mathcal{P}(\{0, \ldots, n - 1\}) \setminus \{\emptyset\} \cong K.
  \]
  However, the monus operation is not a bijection, since
  \[
    0 \dotminus 1 = 1 \dotminus 1 = 0.
  \]
  Restricting to the set of vertices is easier in practice, since one only needs to care about a function's behavior localized to the set of vertices.

  Simplicial complexes are also sometimes defined with an additional axiom stipulating that they are \emph{total}. In this context, one would add the condition\label{total}
\begin{lstlisting}
total : ∀ x : E, {x} ∈ K.faces
\end{lstlisting}
  This avoids the above pathological example, but introduces a much more severe complication for subdivisions. Stellar subdivisions, for example, introduce a new vertex. If $K$ is total, then one either needs to change the underlying type to add a new label for this vertex, or shift the labels to make room for it. In either case, a monumental amount of overhead is generated to handle typecasting. In particular, many identities will only hold up to isomorphism where our current formalization will give true equalities.
\end{remark}

\subsubsection{Simplicial Coercions}
Isomorphisms give a useful tool for attempting to change the underlying type of a simplicial complex without loss of data. However, the definition relies on specifying a simplicial complex as the codomain and an inverse morphism. For practical use cases, it is more convenient to just provide a well-behaved function $\varphi : E \to F$ for typecasting. Moreover, one desires that this typecasting is lossless, essentially behaving as a relabeling of the vertices.
\begin{lstlisting}
structure SimplicialCoe
    (K : AbstractSimplicialComplex E) (F : Type _)
  where
    coe : E → F
    Injective : Set.InjOn coe (K.vertices)
\end{lstlisting}

The core idea is that we lift functions $\varphi : E \to F$ used for typecasting to simplicial maps. Note, $\varphi$ need not be injective everywhere. To have that $\varphi$ is an isomorphism of $K$ onto its image, it suffices to know that $\varphi$ is injective on the set of vertices of $K$.

Some operations on simplicial complexes will change the underlying type $E$. Thus, a reasonable method for typecasting will be indispensable for ensuring complexes have the ``correct'' type at any given time, particularly in the context of joins. In this instance, we find that having a specific construct for typecasting on simplicial complexes to be handy, since we can omit details about codomains and proofs that $\varphi$ is simplicial. Rather, we can just supply a function that is injective on the vertex set.

This localized injectivity condition allows us to prove some useful utility lemmas. For example, simplicial isomorphisms are coercions (\lstinline|SimplicialMap.coe|). Moreover, they are preserved under restriction to subcomplexes (\lstinline|SimplicialCoe.restrict|) and composition with isomorphisms (\lstinline|SimplicialCoeOnImage|). They also distribute over set operations like unions.
\begin{proposition}[{\small\texttt{simplicialCoe_union}}]
  Let $\varphi : E \to F$ be a simplicial coercion on $K \cup L$. Then,
  \[
    \varphi(K \cup L) = \varphi(K) \cup \varphi(L).
  \]
\end{proposition}

Most importantly, coercions yield nice interactions with the constructions in Sections~\ref{sec:operations} and~\ref{sec:stellar-subdivisions}, allowing us to prove useful identities to simplify typecasting steps in proofs. We present some of these utility lemmas as they naturally arise.

\subsection{Subcomplexes \& Operations}\label{sec:operations}
We explore some notable subcomplexes of and operations on a simplicial complex. We focus first on those that will be essential for defining and studying stellar subdivisions in Section~\ref{sec:stellar-subdivisions}. A number of additional constructions; such as stars, cones, and balls around a face; are formalized, with definitions and some simple results available. We then follow with an investigation toward implementing the join of two simplicial complexes.

\subsubsection*{Links}

The first and foremost construction is the \emph{link} of a face $s$:
\begin{lstlisting}
def Link (K : AbstractSimplicialComplex E) (s : Finset E) :=
  ⟨{t ∈ K.faces | s ∪ t ∈ K.faces ∧ s ∩ t = ∅}, _, _⟩
\end{lstlisting}
Given the $n$-face $s$ in $K$, the link $\text{Lk}(K, s)$ is the collection of all faces in $K$ that are disjoint from $s$ and that form a face of $K$ when combined with $s$. Figure~\ref{fig:subcomplexes}a shows an example of a link. Pictorially, $\text{Lk}(K, s)$ represents all the faces in $K$ that ``support'' $s$ via some ``linking'' face. They enjoy nice properties on their own, including interactions with isomorphims and set operations.

\begin{proposition}[{\small\texttt{link_simplicialIso}}]
  Let $f : E \to F$ be an isomorphism between $K$ and $L$. Then,
  \[
    \text{Lk}(K, s) \cong \text{Lk}(L, f(s)).
  \]
\end{proposition}

\begin{proposition}[{\small\texttt{link_simplicialUnion}}, {\small\texttt{link_simplicialInter}}]\phantom{text}
  \begin{itemize}
    \item $\text{Lk}(K \cup L, s) = \text{Lk}(K, s) \cup \text{Lk}(L, s)$.
    \item $\text{Lk}(K \cap L, s) = \text{Lk}(K, s) \cap \text{Lk}(L, s)$.
  \end{itemize}
\end{proposition}

Links also enjoy a useful commutativity property with simplicial coercions.
\begin{proposition}[{\small\texttt{link_simplicialCoe_image}}]
  Let $\varphi : E \to F$ be a simplicial coercion on $K$. Then,
  \[
    \varphi(\text{Lk}(K, s)) = \text{Lk}(\varphi(K), \varphi(s)).
  \]
\end{proposition}

Further, we can derive some useful identities for links of subfaces.
\begin{proposition}[{\small\texttt{link_of_face_complement}}]
  If $t \subseteq s$, then
  \[
    Lk(K, s) = Lk(Lk(K, t), s \setminus t).
  \]
\end{proposition}

\begin{figure}[t]
    \centering
    \includegraphics[width=0.8\linewidth]{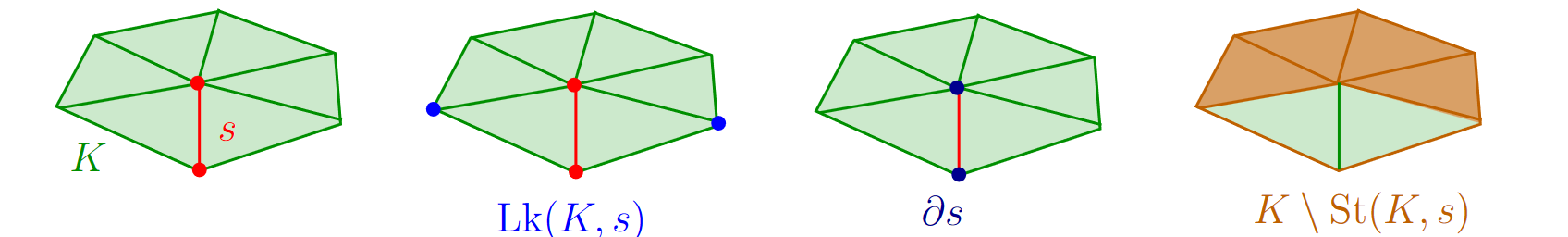}
    \caption{Examples of subcomplexes of $K$ formed by operations on the face $s$. From left to right: (a) the link $\text{Lk}(K, s)$, (b) the face boundary $\partial s$, (c) the star complement $K \setminus \text{St}(K, s)$.}
    \label{fig:subcomplexes}
\end{figure}

\subsubsection*{Face Boundaries}

The \emph{boundary} of an $n$-face $s$, denoted $\partial s$, is the simplicial complex formed by all the $(n - 1)$-faces of $s$. Figure~\ref{fig:subcomplexes}b shows an example geometrically. Described formally,
\begin{lstlisting}
def FaceBoundary (s : Finset E) := ⟨(Finset.powerset s \ {s, ∅}), _, _⟩
\end{lstlisting}

The boundary has useful qualities for studying individual faces and their subfaces. Results on monotonicity properties and simplicial maps are straightforward to obtain, such as preservation under isomorphisms and commuting with coercions.

\begin{proposition}[{\small\texttt{faceBoundary_mem_iff_subset}}]
  $t \in \partial s \iff t \subset s \land t \neq \emptyset$.
\end{proposition}

\begin{proposition}[{\small\texttt{faceBoundary_subface_subcomplex}}]
  If $s \subseteq t$, then $\partial s \subseteq \partial t$.
\end{proposition}

\begin{proposition}[{\small\texttt{faceBoundary_simplicialIso}}]
  If $s \cong t$, then $\partial s \cong \partial t$.
\end{proposition}

\begin{proposition}[{\small\texttt{faceBoundary_simplicialCoe_image}}]
  If $\varphi$ is injective on $s$, then
  \[
    \partial (\varphi (s)) = \varphi (\partial s).
  \]
\end{proposition}

\subsubsection*{Stars \& Star Complements}

The \emph{star complement} of $K$ with respect to $s$ is usually defined as the complement of the interior of the \emph{star} around $s$. Intuitively, the star around $s$ represents a ``neighborhood'' with all faces that contain $s$, denoted by $\text{St}(K, s)$. Thus, the star complement takes $K$ and removes the ``interior'' of the star, denoted by $K \setminus \text{St}(K, s)$. Figure~\ref{fig:subcomplexes}c shows a visual example of the latter.

Regarding first the star around $s$, we can define it formally as:
\begin{lstlisting}
def StarNeighborhood (X : AbstractSimplicialComplex E) (s : Finset E) :=
    ⟨{t ∈ X.faces | s ∪ t ∈ X.faces}, _, _⟩
\end{lstlisting}

As in the case of links, we can provide conditions under which isomorphisms descend to stars.

\begin{proposition}[{\small\texttt{star_simplicialIso}}]
  Let $f : E \to F$ be an isomorphism between $K$ and $L$. Then,
  \[
    \text{St}(K, s) \cong \text{St}(L, f(s)).
  \]
\end{proposition}

Regarding the star complement, we can define it directly for convenience.

\begin{lstlisting}
def StarComplement (K : AbstractSimplicialComplex E) (s : Finset E) :=
  ⟨{t ∈ K.faces | ¬s ⊆ t}, _, _⟩
\end{lstlisting}

Star complements enjoy their own variations of isomorphism and coercion preservation results.

\begin{proposition}[{\small\texttt{starComplement_simplicialIso}}]
  Let $f : E \to F$ be an isomorphism between $K$ and $L$. Then,
  \[
    K \setminus \text{St}(K, s) \cong L \setminus \text{St}(L, f(s)).
  \]
\end{proposition}

\begin{proposition}[{\small\texttt{starComplement_simplicialCoe_image}}]
  Let $\varphi : E \to F$ be a simplicial coercion on $K$. Then,
  \[
    \varphi(K \setminus \text{St}(K, s)) = \varphi(K) \setminus \text{St}(\varphi(K), \varphi(s)).
  \]
\end{proposition}

\subsubsection*{Cones}

The \emph{cone} over a simplicial complex $K$ is defined as the join (to be defined in the following subsection) of $K$ with an additional point $x$, denoted by $\text{Cone}(K, x)$. Figure~\ref{fig:join} shows a visual example.

\begin{lstlisting}
def Cone
    (K : AbstractSimplicialComplex E) (x : E)
    (x_nin_K : x ∉ K.vertices) :=
  (Simplex {x}) ⋆ K
\end{lstlisting}
Here, \lstinline|Simplex| converts a finite set $s$ into the simplicial complex given by its powerset.
\begin{lstlisting}
def Simplex (s : Finset E) := ⟨s.powerset.toSet \ {∅}, _, _⟩
\end{lstlisting}

Cones are preserved under isomorphisms and increase the dimension of a simplicial complex by 1.
\begin{proposition}[{\small\texttt{cone_simplicialIso}}, {\small\texttt{cone_dim}}]\phantom{text}
  \begin{itemize}
    \item Let $f : E \to F$ be an isomorphism between $K$ and $L$. Then,
    \[
      \text{Cone}(K, x) \cong \text{Cone}(L, y)
    \]
    where $x \notin V(K)$ and $y \notin V(L)$.
    \item $\dim \text{Cone}(K, x) = \dim K + 1$.
  \end{itemize}
\end{proposition}

\subsubsection*{Disjoint Unions \& Joins}

Finally, we dedicate considerable attention to the \emph{join} of two simplicial complexes. Geometrically, we consider the join of $K$ and $L$, denoted $K \star L$, as the simplicial complex obtained by fully connecting the faces of $K$ with those in $L$ (cf. Figure~\ref{fig:join}). Described as a set, we consider all pairs $s \in K$ and $t \in L$, then take their disjoint unions. Hence, we first describe an appropriate notion of disjoint union.

\begin{figure}
    \centering
    \includegraphics[width=0.8\linewidth]{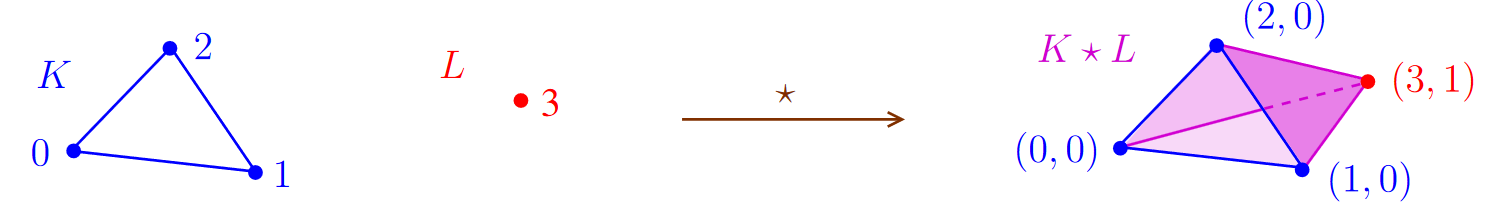}
    \caption{Example of the join of two simplicial complexes, $K$ and $L$. In fact, $K \star L \cong \text{Cone}(K, 3)$.}
    \label{fig:join}
\end{figure}

While \lstinline|mathlib| does include a definition of disjoint union for \lstinline|Finset|, it requires a proof that $s$ and $t$ are disjoint, then returns $s \cup t$. However, it is reasonable to expect that faces of $K$ and $L$ are not always disjoint. We use the common trick of creating disjoint sets of pairs. Stated formally in Lean:
\begin{lstlisting}
variable (𝕜 E) [DecidableEq 𝕜] [DecidableEq E]
variable [Semiring 𝕜] [Nontrivial 𝕜]
def FaceDisjointUnion (s t : Finset E) : Finset (E × 𝕜) :=
  s ×ˢ {0} ∪ t ×ˢ {1}
infixl:65 " ⊔ₛ " => FaceDisjointUnion

def SimplicialJoin (K L : AbstractSimplicialComplex E) :=
  ⟨{(s ⊔ₛ t) | (s ∈ K.faces ∪ {∅}) (t ∈ L.faces ∪ {∅})} \ {∅}, _, _⟩
\end{lstlisting}

\begin{remark}
  The definition deserves some explanation. First, the set description is designed to ensure that $s \sqcup_s \emptyset$ and $\emptyset \sqcup_s t$ are faces, so that $K$ and $L$ naturally embed into the join. Then, we remove $\emptyset \sqcup_s \emptyset = \emptyset$ to satisfy \lstinline|empty_notMem|.

  More striking is the choice of assumptions on $\mathbb{k}$. The basic function of $\mathbb{k}$ is for indexing which side of the disjoint union a face is in. If $\mathbb{k}$ is a nontrivial semiring, then there is a natural choice using 0 and 1 notation, mimicking the usual indexing by $\mathbb{N}$. This choice is designed with \lstinline|SimplicialComplex| and \lstinline|convexJoin| in mind, since if $\mathbb{k}$ is a ring and $E$ is a $\mathbb{k}$-module, then $E \times \mathbb{k}$ will have a natural $\mathbb{k}$-module structure and use identical notation. Hence, this choice sets up future paths to lifting this definition to the geometric case.

  One might alternatively choose to use the sum type $E \oplus E$. This has the advantage of more naturally capturing the disjoint union of types and directly encodes properties of the join like commutativity and associativity. For example, there is a natural equivalence $(E \oplus E) \oplus E \simeq E \oplus (E \oplus E)$. However, the sum type of $\mathbb{k}$-modules need not have a $\mathbb{k}$-module structure itself.
  
  Another choice that does carry a natural $\mathbb{k}$-module structure is the direct sum of a module with itself, which is equivalent to $E \times E$. However, this presents an issue when attempting to capture disjoint unions. For example, consider $K \star K$. One might desire that each copy of $K$ embeds into $E \times \{0\}$ and $\{0\} \times E$, respectively. However, if $0$ is a vertex of $K$, then both inclusions give the point $(0, 0)$ in $E \times E$ rather than distinguishing them as two distinct copies. Therefore, another trick is required to avoid such collisions. We find using $E \times \mathbb{k}$ to be a somewhat parsimonious solution, since it also carries a natural $\mathbb{k}$-module structure and captures the idea of disjoint unions using the standard trick of indexing by 0 and 1.
\end{remark}

Joins are associative and commutative up to typecasting and isomorphism, and moreover preserve isomorphism type. The empty complex also acts as a multiplicative identity.

\begin{proposition}[{\small\texttt{simplicialJoin_assoc}}, {\small\texttt{simplicialJoin_comm}}]\label{prop:join-assoc}\phantom{text}
  \begin{itemize}
    \item Let $\varphi : E \to E$ be a simplicial coercion on $K \star L$ and $\psi : E \to E$ be a simplicial coercion on $L \star Z$. Then,
    \[
      \varphi (K \star L) \star Z \cong K \star \psi (L \star Z).
    \]
    \item $K \star L \cong L \star K$.
  \end{itemize}
\end{proposition}

\begin{proposition}[{\small\texttt{simplicialJoin_simplicialIso}}]
  If $K \cong L$ and $Z \cong W$, then
  \[
    K \star Z \cong L \star W.
  \]
\end{proposition}

\begin{proposition}[{\small\texttt{simplicialJoin_bot_id}}] $K \star \bot \cong \bot \star K \cong K$.
\end{proposition}

The join obeys a distributive law over unions and intersections, as well as observing a monotonicity property with respect to subcomplexes. This is in fact true on the level of equality, not just up to isomorphism.

\begin{proposition}[{\small\texttt{simplicialJoin_simplicialUnion}}]\phantom{text}
  \begin{itemize}
    \item $K \star (L \cup Z) = (K \star Z) \cup (L \star Z)$.
    \item $(K \cup L) \star Z = (K \star Z) \cup (L \star Z)$.
  \end{itemize}
\end{proposition}

\begin{proposition}[{\small\texttt{simplicialJoin_simplicialInter}}]\phantom{text}
  \begin{itemize}
    \item $K \star (L \cap Z) = (K \star Z) \cap (L \star Z)$.
    \item $(K \cap L) \star Z = (K \star Z) \cap (L \star Z)$.
  \end{itemize}
\end{proposition}

\begin{proposition}[{\small\texttt{simplicialJoin_subcomplex}}]
  If $K \subseteq L$ and $Z \subseteq W$, then
  \[
    K \star Z \subseteq L \star W.
  \]
\end{proposition}

Finally, the dimension of a join is additive up to a constant.

\begin{proposition}[{\small\texttt{dim_of_join}}]
  $\dim (K \star L) = \dim K + \dim L + 1$.
\end{proposition}

These proofs often make heavy use of the structure of faces under disjoint unions. Therefore, we provide numerous utility lemmas detailing properties about our definition of disjoint union above.

Since the complex $K \star L$ has the underlying type $E \times \mathbb{k}$, repeated applications of the join operation require typecasting $E \times \mathbb{k} \to E$, as seen in Proposition~\ref{prop:join-assoc} for associativity. In practice, we can often get away with simple maps that will completely cover our use cases in Section~\ref{sec:stellar-subdivisions}. For example, we defined disjoint unions in the manner above to handle the case where $V(K) \cap V(L) \neq \emptyset$. If they are disjoint, then the projection $p_1 : E \times \mathbb{k} \to E$ to the first coordinate will be injective on the vertex sets, hence can be used as a coercion. The faces then have nicer forms to work with, enabling much simpler proofs through some basic rewrites.

\begin{proposition}[{\small\texttt{simplicialJoinProj_mem}}]
  \[
    s \in p_1(K \star L) \iff \exists t \in K \cup \{\emptyset\},\; \exists u \in L \cup \{\emptyset\},\; s = t \cup u \land t \cup u \neq \emptyset.
  \]
\end{proposition}

\begin{proposition}[{\small\texttt{simplicialJoinProj_mem_vertices}}]
  \[
    x \in V(p_1(K \star L)) \iff x \in V(K) \lor x \in V(L).
  \]
\end{proposition}

\begin{figure}
    \centering
    \includegraphics[width=0.85\linewidth]{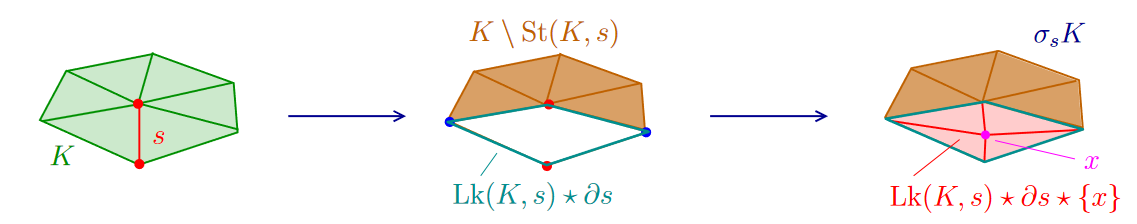}
    \caption{Example of a stellar subdivision, computed in steps. The interior of the star around $s$ is removed and the barycenter $x$ is placed inside, which is then fully connected to the complex via a join to $\text{Lk}(K, s) \star \partial s$.}
    \label{fig:stellar-subdivision}
\end{figure}

Finally, we observe that joins interact nicely with some of the other objects previously defined. For example, joins enjoy distributive properties with links and star complements. Note that these are true equalities.

\begin{proposition}[{\small\texttt{simplicialJoin_link}}]
  \[
    \text{Lk}(K \star L, s \sqcup_s t) = \text{Lk}(K, s) \star \text{Lk}(L, t).
  \]
\end{proposition}

\begin{proposition}[{\small\texttt{simplicialJoin_starComplement_left}}]
  \[
    (K \setminus \text{St}(K, s)) \star L = (K \star L) \setminus \text{St}(K \star L, s \sqcup_s \emptyset).
  \]
\end{proposition}

\section{Stellar Subdivisions}\label{sec:stellar-subdivisions}

Stellar subdivisions subdivide a ``neighborhood,'' or \emph{star}, around a face in a manner that preserves the homeomorphism type. Following Figure~\ref{fig:stellar-subdivision}, we excise out the interior of the star, introduce a new point $x$ (called the \emph{barycenter}), then use joins to connect the barycenter to the star's boundary and fill in the gaps. Informally stated:
\[
    \sigma_s K := K \setminus \text{St}(K, s) \cup \text{Lk}(K, s) \star \partial s \star \{x\},\quad \text{where } x \notin V(K).
\]

Note that we require iterated joins and a union for the definition. For the expression to typecheck in Lean, this means we must typecast $E \times \mathbb{k} \to E$. However, the components of the join are pairwise disjoint. Thus, we can apply $p_1$. Stated formally in Lean:
\begin{lstlisting}
def StellarSubdivision (K : AbstractSimplicialComplex E)
    (s : Finset E) (s_in_K : s ∈ K.faces)
    (x : E) (x_nin_K : x ∉ K.vertices) :
  AbstractSimplicialComplex E :=
    ⟨K\St(K, s) ∪ p₁ ''ˢ (p₁ ''ˢ (Lk(K, s) ⋆ ∂s) ⋆ {x}), _, _⟩
\end{lstlisting}
We can leverage simple rewrite lemmas (as in Section~\ref{sec:operations}) that allow us to act as if these typecasting steps never occurred. We call going from a stellar subdivision $\sigma_s K$ back to $K$ a \emph{stellar weld}. A \emph{stellar move} is then a stellar subdivision, weld, or an isomorphism. One can concretely describe a stellar weld as a set of faces.

\begin{proposition}[{\small\texttt{stellarWeld_faces}}]
  \begin{align*}
    K &= \{t \mid t \in \sigma_s K,\; x \notin t\} \cup \{s \cup (t \setminus \{x\}) \mid t \in \sigma_s K,\; x \in t\} \\
    &= ((\sigma_s K) \setminus \text{St}(\sigma_s K, \{x\})) \cup \{s \cup (t \setminus \{x\}) \mid t \in \sigma_s K,\; x \in t\}.
  \end{align*}
\end{proposition}

As far as the authors are aware, there is no standard reference for this description, so we provide the intuition behind the construction, with details formalized in our code.

\begin{proof}
  There are two cases for any face $t \in \sigma_s K$. Either $x \notin t$, in which case $t \in K \setminus \text{St}(K, s)$, or $x \in t$ and hence $t$ is in the subdivided star. In the first case, $t \in K$, since $K \setminus \text{St}(K, s)$ is a subcomplex of $K$. In the second case (cf. Fig.~\ref{fig:stellar-weld}), we can excise the barycenter point $x$ and glue the face $s$ back in its place. Note that $s$ is recovered by $t = \{x\}$. The second equality follows since $x \in t$ iff $\{x\} \subseteq t$.
\end{proof}

\begin{figure}
    \centering
     \includegraphics[width=0.7\linewidth]{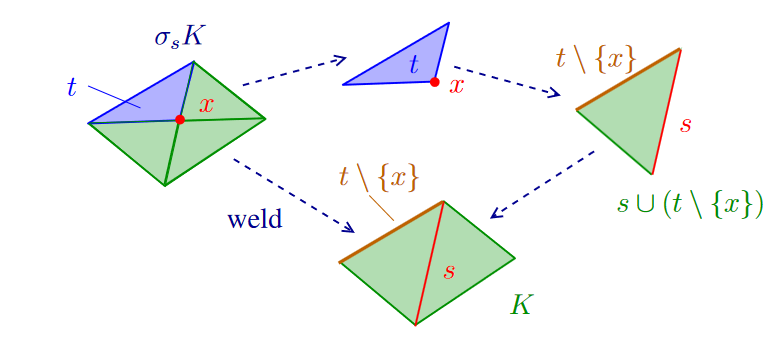}
    \caption{Example of a stellar weld, computing how the face $t \in \sigma_s K$ changes.}
    \label{fig:stellar-weld}
\end{figure}

\subsection{Stellar Equivalence}
We can consider sequences of stellar moves and generate an equivalence relation. We say that $K$ and $L$ are \emph{stellar equivalent} if there exists a finite sequence of stellar moves from $K$ to $L$, denoted $K \cong_\text{st} L$. Formally, this is the transitive closure of the stellar move relation.
\begin{lstlisting}
def StellarEquiv :
  AbstractSimplicialComplex E → AbstractSimplicialComplex E → Prop :=
    Relation.ReflTransGen StellarMove
\end{lstlisting}
This has the advantage that proofs can be straightforwardly done via transitivity or induction. We moreover provide support for quality-of-life tactics like \lstinline|calc|.

One challenge is to consider when typecasting from $E$ to $F$ preserves stellar equivalence. Suppose $K$, $L$ are simplicial complexes over $E$ and $Z$, $W$ over $F$. If $K \cong Z$, $L \cong W$, and $K \cong_\text{st} L$, can we say $Z \cong_\text{st} W$? The answer is: not always. The reason is that one needs to derive a sequence of stellar moves from $Z$ to $W$ given those from $K$ to $L$.

Recall that the definition of a stellar subdivision requires the existence of an unused value $x : E$ to act as the barycenter. Then, for example, if $F$ is a finite type such that there are more values $x_1, \ldots, x_n : E$ used in the relation $K \cong_\text{st} L$ than available in $F$, it may be impossible to directly derive $Z \cong_\text{st} W$. Thus, any result showing that stellar equivalence is preserved after typecasting must have an additional assumption to ensure that this is done without losing information about the barycenters.

\begin{theorem}[\texttt{stellarEquiv_iso}]\label{thm:iso-to-stellar-equiv}
  Let $K$, $L$ be simplicial complexes over $E$ and $Z$, $W$ over $F$ such that $K \cong Z$ and $L \cong W$. Let $f : E \to F$ be injective. Then,
  \[
    K \cong_\text{st} L \implies Z \cong_\text{st} W.
  \]
\end{theorem}

Equivalently, stipulating the existence of an injective $f : E \to F$ means that $|E| \leq |F|$ on the level of cardinalities. This is to ensure that there is always ``room'' in $F$ to support any barycenters associated with the assumption $K \cong_\text{st} L$.

\begin{proof}
  We desire that a sequence of stellar moves over type $E$ induces a sequence of stellar moves over type $F$. If we show the case for a single stellar subdivision or stellar weld, the claim follows by induction. Therefore, it suffices to prove the following lemma.
\end{proof}

\begin{lemma}[\texttt{stellarMove_exists_iso}]
  Under the same assumptions, if $K \cong Z$ and there is a stellar move from $K$ to $L$, then there exists $W$ such that $L \cong W$ and $Z \cong_\text{st} W$.
\end{lemma}
\begin{proof}
  We sketch the argument in the stellar subdivision case. The stellar weld case is similar. The proof can be expressed as a diagram:
  \[\begin{tikzcd}
    K &&&& {\sigma_sK} \\
    \\
    Z && {f(K)} && {f(\sigma_sK) = \sigma_{f(s)}f(K)}
    \arrow["{\text{stellar move}}", from=1-1, to=1-5]
    \arrow["\cong"', from=1-1, to=3-1]
    \arrow["f", from=1-1, to=3-3]
    \arrow["f", from=1-5, to=3-5]
    \arrow["\cong"', from=3-1, to=3-3]
    \arrow[from=3-3, to=3-5]
  \end{tikzcd}\]
  Since $f$ is injective, $K \cong f(K)$. Hence, $f(K) \cong Z$ by transitivity. Likewise, $\sigma_s K \cong f(\sigma_s K)$ by injectivity. It suffices to show that $f(\sigma_s K) = \sigma_{f(s)} f(K)$. Then, $Z \cong_\text{st} \sigma_{f(s)} f(K)$, a stellar subdivision. Setting $W := \sigma_{f(s)}f(K)$ gives the desired result.

  The identity comes by carefully running through the definitions of each simplicial complex and showing both inclusions. The argument is elementary, but lengthy. Details are formalized in \lstinline|stellarSubdivision_injective_image|. In fact, this equality can be strengthened to simplicial coercions.
\end{proof}

\begin{proposition}[{\small\texttt{stellarCoe_stellarSubdivision}}]
  Let $\varphi : E \to F$ be a simplicial coercion on $K$. Then,
  \[
    \varphi(\sigma_s K) = \sigma_{\varphi(s)} \varphi(K).
  \]
\end{proposition}

\subsection{Operations on Stellar Subdivisions}
Finally, we provide a study of how stellar subdivisions interact with some of the operations defined in Section~\ref{sec:operations}. The first result is about how the join of a subdivision can be rewritten as a subdivision of the join. This allows us to obtain results on stellar equivalences and joins.

\begin{proposition}[{\small\texttt{simplicialJoin_stellarSubdivision_left}}]
  \[
    (\sigma_{s}K) \star L = \sigma_{s \sqcup_s \emptyset}(K \star L).
  \]
\end{proposition}

\begin{proposition}[{\small\texttt{simplicialJoin_stellarEquiv}}]
  If $K \cong_\text{st} L$ and $Z \cong_\text{st} W$, then
  \[
    K \star Z \cong_\text{st} L \star W.
  \]
\end{proposition}

We also consider how stellar subdivisions interact with links. This is a core issue in the theory of manifolds, specifically if one wants to study subdivisions on the boundary of an $n$-dimensional combinatorial manifold $M^n$:
\[
  \partial M^n := \{s \in M^n \mid \forall k,\; \text{Lk}(M^n, s) \not\cong_\text{st} S^k\}.
\]
Here, the ``sphere'' $S^k$ is defined as
\[
  S^k := \mathcal{P}(\{0, \ldots, k + 1\}) \setminus \{\{0, \ldots, k + 1\}, \emptyset\}.
\]

Thus, given $t \in \sigma_s K$, we want to consider $\text{Lk}(\sigma_s K, t)$ in terms of the original complex $K$ when possible. In particular, when can we derive an equality, isomorphism, or stellar equivalence? The simplest case is if $t \in K \setminus \text{St}(K, s)$ and $t \notin \text{Lk}(K, s) \star \partial s \star \{x\}$. Geometrically, $t$ is \emph{outside} the neighborhood being subdivided. Thus, the link should be unaffected by a stellar subdivision around the face $s$.

\begin{theorem}[{\small\texttt{stellarSubdivision_link_of_starComplement}}]
  If $t \in K \setminus \text{St}(K, s)$ and $t \notin \text{Lk}(K, s) \star \partial s$, then
  \[
    \text{Lk}(\sigma_s K, t) = \text{Lk}(K, t).
  \]
\end{theorem}

The assumption $t \notin \text{Lk}(K, s) \star \partial s$ is sufficient, since one can think of $\text{Lk}(K, s) \star \partial s$ as the boundary of the star being subdivided (cf. Figure~\ref{fig:stellar-subdivision}). Equivalently,

\begin{theorem}[{\small\texttt{stellarSubdivision_link_of_barycenter}}, {\small\texttt{star_boundary_is_join}}]
  \[
    K \setminus \text{St}(K, s) \cap \text{Lk}(K, s) \star \partial s \star \{x\} = \text{Lk}(K, s) \star \partial s = \text{Lk}(\sigma_s K, \{x\}).
  \]
\end{theorem}

\begin{figure}
    \centering
     \includegraphics[width=0.8\linewidth]{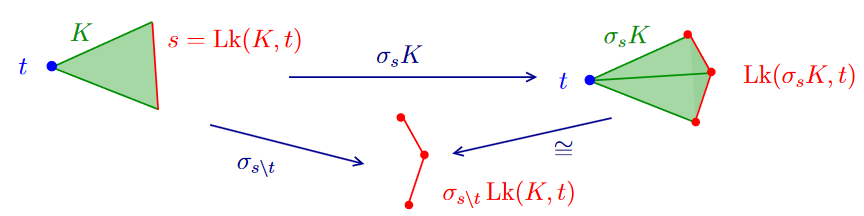}
    \caption{Geometric example of the identity stated in Theorem~\ref{thm:anticomm}.}
    \label{fig:link-subdivision}
\end{figure}

In the case where $t \in \text{Lk}(K, s) \star \partial s$, we obtain a more complicated relation. In general, $\text{Lk}(\sigma_s K, t)$ will contain $x$, thus any relationship to $\text{Lk}(K, t)$ cannot be an equality. Indeed, it need not be an isomorphism either. However, we can derive a stellar equivalence, in which the link and stellar subdivision ``anticommute'' in a precise sense (cf. Figure~\ref{fig:link-subdivision}).

\begin{theorem}[{\small\texttt{stellarSubdivision_anticomm_link}}]\label{thm:anticomm}
  If $t \in \text{Lk}(K, s) \star \partial s$, then
  \[
    \text{Lk}(\sigma_s K, t) = \sigma_{s \setminus t}(\text{Lk}(K, t)).
  \]
\end{theorem}

In all cases, these identities are proven directly by unfolding the definitions of both sets in the equality. This results in rather lengthy arguments (the previous theorem alone constitutes $\sim$2,000 lines of code), but yields direct, elementary proofs. This perhaps demonstrates why some of this theory has resolved into folklore status, since arguments can involve an immense amount of tedium. In fact, to our knowledge, the previous theorem has no standard reference in the literature. We run through a simple portion of the argument to demonstrate the general flow of these proofs and provide a more complete reference for the above.

\begin{proof}
  We show the backwards inclusion. Let $u \in \sigma_{s \setminus t} (\text{Lk}(K, t))$. By definition,
  \[
    u \in \text{Lk}(K, t) \setminus \text{St}(\text{Lk}(K, t), s \setminus t) \lor u \in \text{Lk}(\text{Lk}(K, t), s \setminus t) \star \partial (s \setminus t) \star \{x\}.
  \]
  To obtain $u \in \text{Lk}(\sigma_s K, t)$, we must show $u \in \sigma_s K$, $t \cup u \in \sigma_s K$, and $t \cap u = \emptyset$.

  If $u \in \text{Lk}(K, t) \setminus \text{St}(\text{Lk}(K, t), s \setminus t)$, then unfolding the definition gives $u$, $t \cup u \in K$; $t \cap u = \emptyset$; and $s \setminus t \not\subseteq u$. Note that $u$, $t \cup u \in K \setminus \text{St}(K, s)$ since $u$, $t \cup u \in K$ and since $s \setminus t \not\subseteq u$ implies both $s \not\subseteq u$ and $s \not\subseteq t \cup u$. Hence, $u \in \text{Lk}(\sigma_s K, t)$.

  If $u \in \text{Lk}(\text{Lk}(K, t), s \setminus t) \star \partial (s \setminus t) \star \{x\}$, then $u = u_1 \cup u_2 \cup u_3$ mutually disjoint and not all empty where $u_1 \in \text{Lk}(\text{Lk}(K, t), s \setminus t)$, $u_2 \in \partial (s \setminus t)$, $u_3 \in \{x\}$. Likewise, we can split $t = t_1 \cup t_2$ where $t_1 \in \text{Lk}(K, s)$ and $t_2 \in \partial s$. We can derive by unfolding the definitions:
  \begin{enumerate}
    \item $s \setminus t_1 = s$
    \item $s \setminus t_2 \neq \emptyset$
    \item If $u_1 = \emptyset$ or $u_2 \cup u_3 = \emptyset$, then $s \not\subseteq u$.
  \end{enumerate}

  If $u_2 \cup u_3 = \emptyset$, then $u \in \text{Lk}(\text{Lk}(K, t), s \setminus t)$. The definitions and above properties yield $u$, $t \cup u \in K \setminus \text{St}(K, s)$. $t \cap u = \emptyset$ follows since $u \in \text{Lk}(K, t)$. Therefore, $u \in \text{Lk}(\sigma_s K, t)$.
  
  Otherwise, suppose $u_1 = \emptyset$. Then, $\partial (s \setminus t) \star \{x\} \subseteq \partial s \star \{x\}$ implies $u = u_2 \cup u_3 \in \sigma_s K$. Properties (1) and (2) give $t_2 \cup u_2 \in \partial s$, hence $t \cup u \in \sigma_s K$. $x \notin t$, so $t \cap u_3 = \emptyset$. Moreover, $u_2 \subset s \setminus t$ implies $t \cap u_2 = \emptyset$. Therefore, $t \cap u = \emptyset$ and $u \in \text{Lk}(\sigma_s K, t)$.

  If $u_1 \neq \emptyset$, then $u \in \sigma_s K$ follows by combining the above arguments. $t \cup u \in \sigma_s K$ follows by splitting $t \cup u = (t_1 \cup u_1) \cup (t_2 \cup u_2) \cup u_3$. Similar arguments give that $t_1 \cup u_1 \in \text{Lk}(K, s)$, $t_2 \cup u_2 \in \partial s$, and $u_3 \in \{x\}$. Likewise, $t \cap u = \emptyset$ follows since $t \cap u_1 = \emptyset$ by definition then running the previous argument to get $t \cap u_2 = t \cap u_3 = \emptyset$. Therefore, $u \in \text{Lk}(\sigma_s K, t)$.
\end{proof}

\section{Discussion}\label{sec:discussion}

\subsection{Design \& Challenges}
This project stemmed from a desire to see a formalized proof of Alexander and Pachners' theorems on (bi)stellar subdivisions. In particular, to validate the fine details, primarily following the purely combinatorial proof by Hannes~\cite{friedl2024pachner}. The first necessary step is a useful formalization of simplicial complexes and stellar subdivisions. Throughout, we were motivated by three main points: (a) applications to manifolds, (b) maintaining computability, and (c) maintaining generality. The first point motivated our choice of results on stellar subdivisions. The second point is motivated by a downstream desire to use subdivisions for concrete computations in applications. The third point arose by testing the assumptions and conclusions of Hannes' results. In many cases, we were able to remove conditions like finiteness or strengthen results to equality instead of isomorphism.

\subsubsection{Simplicial Complexes}
This project began before the addition of \lstinline|SimplicialComplex| to \lstinline|mathlib|. Initially, our goal was to reformulate our results, but eventually pivoted to isolating just abstract simplicial complexes, since supporting the extra geometric theory was too large of a detour for our purposes.

We highlight some differences in conventions between our original work and the current formalization. First, we relax the assumptions on the labeling type $E$. Since we live in the setting of purely combinatorial objects, the additional requirement of being a module over some ring $\mathbb{k}$ is unnecessary. We foresee this making work connecting \lstinline|mathlib|'s formalization of the simplex category (defined over $\mathbb{N}$) and undirected graphs (defined over arbitrary vertex type $V$) as 1-dimensional complexes far more straightforward. Indeed, we demonstrate the latter in Remark~\ref{rem:graphs}. We set up future work by inducing instances of \lstinline|SimplicialComplex|, forwarding the results presented here to the geometric case.

Second, we originally allowed $\emptyset$ as a $(-1)$-dimensional face. This is simply a matter of convention---all of our results still hold regardless of which is used. Our original choice simply made much shorter proofs. For example, as seen in Theorem~\ref{thm:anticomm}, faces in
\[
  \text{Lk}(K, s) \star \partial s \star \{x\}
\]
split into unions from each component of the join where at least one is nonempty, generating a fair amount of case work that $\emptyset$ can normally shortcut. We align our convention with \lstinline|mathlib|'s \lstinline|SimplicialComplex| for the purpose of compatibility, which excludes $\emptyset$ as a face.

\subsubsection{Pre-abstract Simplicial Complexes}
Since the preparation of this manuscript, \lstinline|mathlib| has accepted independent formalizations of (pre-)abstract simplicial complexes under the \lstinline|AlgebraicTopology| library. Their definition of \lstinline|PreAbstractSimplicialComplex| is identical to our formalization presented here. They provide instances of partial order and lattice structures, as well as the complex obtained as the image of a map $E \to F$. Additionally, they define \lstinline|AbstractSimplicialComplex| as an extension of \lstinline|PreAbstractSimplicialComplex| by the additional totality axiom:
\begin{lstlisting}
singleton_mem : ∀ v : E, {v} ∈ faces
\end{lstlisting}

We alluded to the issues from including a totality axiom in Remark~\ref{rem:isomorphisms}. Formalizing subdivisions often requires obtaining fresh labels for new vertices generated in the process. This is true for barycentric, stellar, and chromatic subdivisions, for example. Therefore, formalizing subdivisions for both objects will require a radically different approach. We outline an approach that is applicable for \lstinline|PreAbstractSimplicialComplex| in this paper, but highlight some challenges as applies to their \lstinline|AbstractSimplicialComplex|.

Notably, if all values are currently used as labels, then one must find a source of fresh labels either within the existing type or from a new type. In the former case, one must impose a cardinality assumption that $E$ is infinite and supply a map to shift labels so that a value is freed up. In the latter case, one can define the subdivision as being labeled in a new type that contains fresh values. However, these both conceivably generate an immense overhead in handling proper typecasting.

A downstream consequence is that many of the identities proven above are no longer equalities, but rather only hold up to isomorphism. Even if the base type $E$ remains unchanged, this occurs with the added data of an injection $f : E \to E$ to obtain new labels. For example, consider $K \setminus \text{St}(K, s)$ as a subcomplex of both $K$ and $\sigma_s K$ where we use some injective $f$ to label the barycenter. Under this formalization, these two objects are not equal, since $f$ changes the labels of vertices. They are merely isomorphic.

Likewise, one often works with finite abstract simplicial complexes labeled over $\mathbb{N}$ in practice. In this setting, there is an infinite selection of fresh labels for subdivisions to choose from. However, a totality axiom will require a simplicial complex with $n$ vertices to use the type \lstinline|Fin| $n$ and typecast to some \lstinline|Fin| $(n + m)$ whenever $m$ vertices are added by some subdivision.

\subsubsection{Simplicial Maps}
To suitably define simplicial maps for \lstinline|SimplicialComplex|, one needs to carefully preserve the affine independence and gluing conditions. There are many equivalent definitions of simplicial maps, most of which either directly state it for embedded complexes or state it in the abstract case and extend to the embedded case. We believe following the latter would be the more successful approach, but leave some concerns for consideration.

Namely, one needs to extend a simplicial map $f : E \to F$ to be linear on the faces. Obviously, this is true if one simply requires $f : E \to_l\!\![\mathbb{k}]\; F$ to be linear, but this misses many simplicial maps. Hence, a more refined implementation is needed to capture the entire category of simplicial complexes. We judged this complication as too far beyond our interests to resolve at the moment. Particularly, proving that maps are simplicial makes up a substantial portion of our code. Proving linearity conditions would at least double this work for no benefit. Hence, we find it better to leave the abstract definition in a form where one can later provide the correct extension to the geometric case.

Finally, we note that one could, in general, implement isomorphisms to utilize \lstinline|Equiv| or \lstinline|EquivLike| from \lstinline|mathlib|. For example, one may then easily formalize the automorphism group of a simplicial complex. However, both expect the maps to be equivalences of the underlying types. This poses a problem, as we explicitly do not require simplicial isomorphisms to have inverses on the entire ambient type, but rather only on the vertex sets of the respective simplicial complexes. One might resolve this issue with \lstinline|PartialEquiv|, which only requires maps to be inverses on specified source and target sets. This still presents a subtle issue. \lstinline|PartialEquiv| (and, in fact, \lstinline|Equiv|), expects a function between types. So, one would desire some $f : K \to L$ where $K$ and $L$ are fixed (abstract) simplicial complexes. This is doable if one can express something of the approximate form
\begin{lstlisting}
{E F K L : Type _}
[AbstractSimplicialComplex E K] [AbstractSimplicialComplex F L]
(f : K → L) [SimplicialMap K L f]
\end{lstlisting}
similar to the syntax for rings and groups, for example.

However, this requires that \lstinline|AbstractSimplicialComplex| be defined with the \lstinline|class| keyword, rather than \lstinline|structure|. We choose to stay compatible with \lstinline|mathlib|'s implementations, which choose to define simplicial complexes via \lstinline|structure|. In reality, we would desire some \lstinline|PartialEquivLike| feature to fully resolve this issue, which currently does not exist.

\subsubsection{Stellar Equivalence}
Recall that our definition of stellar equivalence was defined between two simplicial complexes over type $E$, which resulted in a fair amount of work to support typecasting between stellar equivalences. See, for example, Theorem~\ref{thm:iso-to-stellar-equiv}. In principle, one could define stellar equivalence between simplicial complexes of different underlying types to bake this property into the definition.

The trade-off is a lack of support for heterogenous equivalence relations in \lstinline|mathlib|. Specifically, relations between different types are given by \lstinline|Rel|, which generalizes the usual \lstinline|Relation|. While we can prove analogs of reflexivity, symmetry, and transitivity, we lose access to any reasonable pre-implemented support for useful tools like induction. We briefly searched for alternative solutions and reached out to the \lstinline|mathlib| community regarding this and the above issue of isomorphisms, but received no fruitful responses. So, we chose to limit our definitions to access the more complete toolbox of \lstinline|Relation|.

\subsection{Related Work}
Within the \lstinline|mathlib| community, significant progress has been made with the formalization of simplicial sets~\cite{cnossen2021simplicial} and the simplex category, including the recent addition of simplicial homology defined for simplicial sets. There have been a number of impressive nontrivial results formalized using this theory, including work on group cohomology~\cite{livingston2023group}, homotopy type theory~\cite{du2025formalization}, higher category theory~\cite{carneiro2025formalizing}, and Quillen equivalences of model categories~\cite{riou2024topcat}. ACL2 formalized simplicial sets and objects in the mid-2000's~\cite{andres2007formalizing} with subsequent work toward verifying the algorithms used in the Kenzo software~\cite{heras2010proving, martin2009acl2}, implementing algorithms that arose from the field of constructive algebraic topology~\cite{rubio2002constructive}. For example, computing homotopy groups of simply connected simplicial sets~\cite{sergeraert1994computability}.

Generalizing to CW complexes, \lstinline|mathlib| recently saw preliminary work formalizing the definition and basic properties~\cite{scholz2024formalisation}. Much earlier, CW complexes had been formalized in HOL Light~\cite{harrison2018hol} and Agda~\cite{buchholtz2018cellular} with the primary intent of formalizing singular homology.

The most directly comparable work to our own is L\"{o}h's~\cite{loh2022exploring} expository text on Lean 3, which uses simplicial complexes as a case study. We found this work to be a useful first reference, but which required universal retooling. However, the text contains many great explicit examples of simplicial complexes and theorems not found in our current work, hence potentially good targets for further development. With the goal of formalizing persistent homology~\cite{heras2013computing, heras2012towards}, simplicial complexes have been formalized in the Rocq proof assistant~\cite{heras2011incidence}. More recently, they have also been formalized in Isabelle~\cite{aransay2021simplicial} for the purpose of studying Alexander duals and evasiveness properties~\cite{aransay2022formalizing}. These works are highly comparable to our own, focusing their attention on abstract simplicial complexes as an explicit computational tool.

\subsection{Future Work}
This work is an initial step toward a larger long-term goal of formalizing Alexander and Pachners' theorems. Much of the foundation for the work presented resulted from a push to define \emph{combinatorial manifolds}~\cite{friedl2019algebraic, friedl2024pachner, lickorish1999simplicial} and to provide an equally comprehensive treatment of stellar subdivisions and homeomorphism type thereon. However, our work building up to and including stellar subdivisions became mature enough to be of its own general interest to share our experiences formalizing them in Lean and inspire further applications. We present a few potential routes inspired by classical use cases or experiences during development.

\subsubsection{Algorithmic Applications}
Our primary downstream consideration is the computational power of simplicial complexes and subdivisions. In particular, simplicial complexes naturally arise as a combinatorial tool for extracting properties of different spaces. Most recognizable are the nerve of a cover, the flag complex of a graph, or the Rips complex of a metric space. One may also consider determining the homotopy type of a space via collapsing algorithms. Discrete Morse theory and shelling orders of simplicial complexes are a key tool in this study. The latter is a core ingredient in Hannes' proof of Pachner's theorem~\cite{friedl2024pachner}. Simplicial homology and cohomology are, of course, natural candidates for future formalization. Kozlov~\cite{kozlov2008combinatorial,kozlov2021organized} provides an excellent overview of these topics and their applications.

\subsubsection{Integration with \texttt{mathlib}}
Given the fragmentary related work already conducted in \lstinline|mathlib|, we placed some emphasis on attempting integration with these various constructions. In particular, isolating just the abstract necessities and type signatures that are as general as possible. Our hope is that our formalization can be somewhat painlessly integrated with these formalizations, particularly with simplicial sets and inducing instances of a geometric \lstinline|SimplicialComplex|, a CW complex, or a graph structure. Indeed, we have started work on resolving issues presented above toward the geometric case, and have initiated plans for integration with elements of the \lstinline|mathlib| library.

\lstinline|Analysis.Convex| also defines the \lstinline|convexJoin| between two sets as the collection of line segments between them. In principle, one could use this as a connection to give a coherent notion of joins for \lstinline|SimplicialComplex|. Formally, one would prove a theorem of the form
\begin{lstlisting}
(K ⋆ L).space = convexJoin 𝕜 K.space L.space
\end{lstlisting}

Moreover, many of the proofs involved in this work are direct, burning through tedious checks that faces satisfy conditions for equalities (cf. Theorem~\ref{thm:anticomm}) or showing that maps are simplicial. Many arguments run similar courses, which could be automated or refactored in more paradigmatic ways, perhaps with more cleverly structured arguments. Indeed, there are instances where we scrapped existing pen-and-paper proofs in favor of more streamlined ones. However, further work can be done to dramatically reduce much of the tedium and proof lengths.

\subsection{Conclusion}
This paper presents a summary of the core definitions and results from our efforts to formalize abstract simplicial complexes and stellar subdivisions. Our motivation and focus is toward computational downstream use cases, in particular the combinatorial theory of manifolds, and isolating only the necessities to achieve that goal, while encouraging future work toward integration with the wide array of similar developments in \lstinline|mathlib|. We constrain ourselves to discussion of the more pertinent features. The available Lean code contains more constructions and results than what can be presented in a clear and concise manner here. While somewhat ``classical,'' the tools of combinatorial topology are not only widely used in many fields of mathematics but also outside of mathematics. We are optimistic that a proper treatment of the subject will end with a useful tool for future formalizations in Lean, hopefully avoiding further slides into mathematical folklore.

\bibliographystyle{plainurl}
\bibliography{refs}
\end{document}